\documentclass[runningheads,a4paper]{llncs}

\usepackage{amssymb}
\usepackage{graphicx}

\usepackage{url}
\urldef{\mailsa}\path|{alfred.hofmann, ursula.barth, ingrid.haas, frank.holzwarth,|
\urldef{\mailsb}\path|anna.kramer, leonie.kunz, christine.reiss, nicole.sator,|
\urldef{\mailsc}\path|erika.siebert-cole, peter.strasser, lncs}@springer.com|    
\newcommand{\keywords}[1]{\par\addvspace\baselineskip
\noindent\keywordname\enspace\ignorespaces#1}
\newcommand{\eat}[1]{}

\newtheorem{assump}{Assumption}
\usepackage{epstopdf}
\usepackage{amssymb,amsmath,amsfonts}
\usepackage{bbm}
\usepackage[]{graphicx}
\usepackage[bottom]{footmisc}
\usepackage{subfigure}
\usepackage{times}
\usepackage{hyperref}
\usepackage{url}
\usepackage{algorithm,algorithmic}
\usepackage{verbatim}
\usepackage{color}
\usepackage{url}
\usepackage{booktabs} 
\usepackage{array} 
\usepackage{paralist} 
\usepackage{bbm}

\begin{document}

\mainmatter  

\title{A Time and Space Efficient Algorithm \\for Contextual Linear Bandits} 

\titlerunning{A Time and Space Efficient Algorithm for Contextual Linear Bandits}

%
%
\author{Jos\'e Bento$^{1}$ \and Stratis Ioannidis$^{2}$  \and S. Muthukrishnan$^{3}$  \and Jinyun Yan$^{3}$}
\authorrunning{Jos\'e Bento \and Stratis Ioannidis  \and S. Muthukrishnan  \and Jinyun Yan}

\institute{$^{1}$Stanford University,  jbento@stanford.edu\\$^{2}$Technicolor, stratis.ioannidis@technicolor.com\\$^{3}$Rutgers University, $\{$muthu, jinyuny$\}$@cs.rutgers.edu  
 }
%
%
\newcommand{\naturals}{\ensuremath{\mathbbm{N}}}
\newcommand{\iden}{\ensuremath{\mathbbm{I}}}
\newcommand{\reals}{\ensuremath{\mathbbm{R}}}
\newcommand{\xset}{\ensuremath{\mathbbm{X}}}
\newcommand{\expect}{\ensuremath{\mathbbm{E}}}
\newcommand{\id}{\ensuremath{\mathbbm{1}}}
\newcommand{\argmax}{\operatornamewithlimits{arg\,max}}
\newcommand{\snote}[1]{\textcolor{blue}{#1}}

\newcommand{\junk}[1]{}
\newcommand{\myset}[1]{\mathcal{#1}}
\newcommand{\T}{\mathcal{T}}
\newcommand{\TaT}{\T_{a,T}}
\newcommand{\Tat}{\T_{a,t}}
\newcommand{\Tatm}{\T_{a,t-1}}

\newcommand{\Gar}{\Gamma_{a,r}}
\newcommand{\Gat}{\Gamma_{a,t}}
\newcommand{\Gan}{\Gamma_{a,n}}

\newcommand{\prob}{\mathbb{P}}

\newcommand{\E}{\mathbb{E}}

\newcommand{\TbT}{\T_{b,T}}
\newcommand{\Tbt}{\T_{b,t}}
\newcommand{\asXt}{a^*_{X_t}}

\newcommand{\Tast}{\T_{\asXt,t}}
\newcommand{\Tastm}{\T_{\asXt,t-1}}

\newcommand{\G}{\mathcal{G}}
\newcommand{\sS}{\mathcal{S}}
\newcommand{\A}{\mathcal{A}}
\newcommand{\X}{\mathcal{X}}
\newcommand{\cH}{\mathcal{H}}

\maketitle
\vspace*{-0.4cm}
\begin{abstract}
We consider a multi-armed bandit problem where payoffs 
are a linear function of an observed stochastic contextual variable.
In the scenario where there exists a gap between optimal and suboptimal rewards,
several algorithms have been proposed that achieve $O(\log T)$ regret after $T$ time steps.
However, proposed methods either have a computation complexity per iteration that 
scales linearly with $T$ or achieve regrets that grow linearly with the number of contexts $|\myset{X}|$.
We propose an $\epsilon$-greedy type of algorithm that solves both limitations.
In particular, when contexts are variables in $\reals^d$,
we prove that our  algorithm
has a constant computation complexity per iteration of $O(poly(d))$ and
can achieve a regret of $O(poly(d) \log T)$
even when $|\myset{X}| = \Omega (2^d) $.
In addition, unlike previous algorithms,
its space complexity scales like $O(Kd^2)$ and does not grow with $T$.
\keywords{ Contextual Linear Bandits, Space and Time Efficiency}
\end{abstract}

\section{Introduction}

The contextual multi-armed bandit problem is a sequential learning problem~\cite{langford2007epoch,dudik2011efficient}.  At each time step, a learner has to chose among a set of possible actions/arms $\myset{A}$. Prior to making its decision, the learner observes some additional side information $x\in\myset{X}$ over which he has no influence. This is commonly referred to as the \emph{context}.
In general, the reward of a particular arm $a\in \myset{A}$ under context $x\in\myset{X}$ follows some unknown distribution. The goal of the learner is to select arms so that it minimizes its expected \emph{regret}, \emph{i.e.}, the expected difference between its cumulative reward and the reward accrued by an optimal policy, that knows the reward distributions.

Langford and Zhang~\cite{langford2007epoch} propose an algorithm called \emph{epoch-Greedy} for general contextual bandits. Their algorithm  achieves an $O(\log T)$ regret in the number of timesteps $T$ in the \emph{stochastic} setting, in which contexts are sampled from an unknown distribution in an i.i.d.~fashion. 
Unfortunately, the proposed algorithm and subsequent improvements \cite{dudik2011efficient} have high computational complexity.  Selecting an arm at time step $t$ requires making a  number of calls to a so-called \emph{optimization oracle} that grows polynomially in $T$. In addition, the cost of an implementation of this optimization oracle can grow linearly in $|\myset{X}|$ in the worst case; this is prohibitive in many interesting cases, including the case where $|\myset{X}|$ is exponential in the dimension of the context. In addition, both algorithms proposed in \cite{langford2007epoch} and \cite{dudik2011efficient} require keeping a history of observed contexts and arms chosen at every time instant. Hence, their space complexity grows linearly in $T$. 

In this paper, we show that the challenges above can be addressed  when rewards are  linear. 
In the above contextual bandit set up, this means that $\myset{X}$ is a subset of $\reals^d$, and the expected reward of an arm $a\in \myset{A}$ is an unknown linear function of the context $x$, \emph{i.e.}, it has the form  $x^\dagger\theta_a$, for some unknown vector $\theta_a$.
This is a case of  great interest, arising naturally when, conditioned on $x$, rewards from different arms are uncorrelated:

\medskip
\noindent
{\bf Example 1.} {\em (Processor Scheduling)}
 A simple example is assigning incoming jobs to a set  of processors $\myset{A}$, whose processing capabilities are not known {\em  a priori}. This could be the case if, \emph{e.g.}, the processors are machines in the cloud or, alternatively, humans offering their services through, \emph{e.g.}, Mechanical Turk. Each arriving job is described by a set of attributes $x\in \reals^d$, each capturing the work load of different types of sub-tasks this job entails, \emph{e.g.},~computation, I/O, network communication, \emph{etc.} Each processor's unknown feature vector $\theta_a$ describes its processing capacity,  \emph{i.e.}, the time to complete a sub-task unit, in expectation. The expected time to complete a task $x$ is given by  $x^\dagger \theta_a$; the goal of minimizing the delay (or, equivalently, maximizing its negation) brings us in the contextual bandit setting with linear rewards.
\hfill $\square$

\medskip
\noindent
{\bf Example 2.} {\em (Display Ad Placement)}
In the online ad placement problem, online users are visiting a website, which must decide which ad to show them selected from a set $\myset{A}$. Each online user visiting the website is described by a set of attributes $x\in\reals^d$ capturing, \emph{e.g.}, its geo-location, its previous viewing history, or any information available through a tracking service like BlueKai.  Each ad $a\in\myset{A}$ has a probability of being clicked that is of the form $x^\dagger \theta_a$, where $\theta_a\in\reals^d$ an unknown vector describing each ad. The system objective is to maximize the number of clicks, falling again under the above contextual bandit setting.
\hfill $\square$

\medskip
\noindent
{\bf Example 3.} {\em (Group Activity Selection)}
Another motivating example is maximizing group satisfaction, observed as the outcome of a secret ballot election. In this setup, a subset of $d$ users congregate to perform a joint activity, such as, \emph{e.g.}, dining, rock climbing,  watching a movie, \emph{etc.} The group is dynamic and, at each time step, the vector $x\in \{0,1\}^d$, is an indicator of  present participants. An arm (\emph{i.e.}, a joint activity) is selected; at the end of the activity, each user votes whether they liked the activity or not in a secret ballot, and the final tally is disclosed. In this scenario, the unknown vectors $\theta_a\in \reals^d$ indicate the probability a given participant will enjoy activity $a$, and the goal is to select activities that maximize the aggregate satisfaction among participants present at the given time step.
\hfill $\square$


Our contributions are as follows.
\begin{itemize}
\item
We isolate and focus on linear payoff case of stochastic multi-armed bandit problems, and design a simple arm selection policy which does not recourse to sophisticated oracles inherent in prior work. 
\item
We prove that our policy achieves an $O(\log T)$ regret after $T$ steps in the stochastic setting,  when the expected rewards of each arm are well separated.  This meets the regret bound of best known algorithms for contextual multi-armed bandit problems. In addition, for many natural scenarios, it scales as $O(poly(d) \log T)$, which we believe
we are the first to prove under arm separation and for an efficient algorithm.
\item
We show that our algorithm has $O(|\A| d^3 )$ computational complexity per step and its expected space complexity scales like $O(| \A| d^2  )$. For algorithms that achieve similar
regrets, this is a significant improvement over known contextual multi-armed bandit problems, as well as for bandits specialized for linear payoffs. 
\end{itemize}

Our algorithm is inspired by the work of \cite{auer2002finite} on the $\epsilon$-greedy algorithm and the use of linear regression to estimate the parameters $\theta_a$. The main technical innovation is the use of matrix concentration bounds to control the error of the estimates of $\theta_a$ in the stochastic setting. We believe that this is a powerful realization and may ultimately help us analyze richer classes of payoff functions.

The remainder of this paper is organized as follows: in Section \ref{sec:related} we compare our results with existing literature. In Section \ref{sec:setup} we describe the set up of our problem in more detail. In Section \ref{sec:main_results} we state our main results and prove them in Section  \ref{sec:proofs}.  Section \ref{sec:numerical} is devoted to exemplifying the performance and limitations of our algorithm by means of simple numerical simulations. We discuss challenges in dealing with an adversarial setting in Section~\ref{sec:discussion}  and draw our conclusions in Section \ref{sec:conclusion}.

\section{Related Work} \label{sec:related}

The original paper by Langford and Zhang~\cite{langford2007epoch} assumes that the context $x\in \myset{X}$ is sampled from a probability distribution $p(x)$ and that, given an arm $a \in \myset{A}$, and conditioned on the context $x$, rewards $r$ are sampled from a probability distribution $p_a(r \mid x)$. As is common in bandit problems, there is a tradeoff between \emph{exploration},  \emph{i.e.}, selecting arms to sample rewards from the distributions $\{p_a(r\mid x)\}_{a\in \myset A}$ and learn about them, and \emph{exploitation}, whereby knowledge of these distributions based on the samples is used to select an arm that yields a high payoff. 

In this setup, a significant challenge is that, though contexts $x$ are sampled independently, they are not independent conditioned on the arm played: an arm will tend to be selected more often in contexts in which it performs well.  Hence, learning the distributions $\{p_a(r\mid x)\}_{a\in \myset{A}}$ from such samples is difficult. 
The epoch-Greedy algorithm~\cite{langford2007epoch} deals with this by separating the exploration and exploitation phase, effectively selecting an arm uniformly at random at certain time slots (the exploration ``epochs''), and using samples collected only during these epochs to estimate the payoff of each arm in the remaining time slots (for exploitation). 
Our algorithm uses the same separation in ``epochs''.
Langford and Zhang \cite{langford2007epoch} establish an $O(T^{2/3}(\ln |\myset{X}|)^{1/3})$ bound on the regret for epoch-Greedy in their stochastic setting. They further improve this to $O(\log T)$ when  a lower bound on the gap between optimal and suboptimal arms in each context exists, \emph{i.e.}, under \emph{arm separation}.

Unfortunately, the price of the generality of the framework in \cite{langford2007epoch} is the high computational complexity when selecting an arm during an exploitation phase. In a recent improvement \cite{dudik2011efficient}, this computation requires a $poly(T)$ number of calls to an optimization oracle. Most importantly, even in the linear case we study here, there is no clear way to implement this oracle in sub-exponential time in $d$, the dimension of the context. As  Dudik \emph{et al.} \cite{dudik2011efficient} point out, the optimization oracle solves a so-called cost-sensitive classification problem. In the particular case of linear bandits, the oracle thus reduces to finding the ``least-costly'' linear classifier. This is hard, even in the case of only two arms: finding the linear classifier with the minimal number of errors is NP-hard \cite{johnson1978}, and remains NP hard even if an approximate solution is required \cite{bartlett99}.  As such, a different approach is warranted under linear rewards. 

Contextual bandits with linear rewards is a special case of the classic linear bandit setup~\cite{auer2002tradeoffs,chu2011contextual,li2010contextual,rusmevichientong2008linearly}. In this setup, the arms themselves are represented as vectors, \emph{i.e.}, $\myset{A}\subset \reals^d$, and, in addition, the set $\myset{A}$ can change from one time slot to the next. The expected payoff of an arm $a$ with vector $x_a$ is given by $x_a^\dagger \theta$, for some unknown vector $\theta\in\reals^d$, \emph{common among all arms}. 

There are several different variants of the above linear model. Auer \cite{auer2002tradeoffs}, Li \emph{et al.}~\cite{li2010contextual}, and
Chu \emph{et al.} \cite{chu2011contextual}, and Li a study this problem in the adversarial setting, assuming  a finite number of arms $|A|$. In the adversarial setting, contexts are not sampled i.i.d. from a distribution but can be an arbitrary sequence, for example, chosen by an adversary that has knowledge of the algorithm and its state variables.
Both algorithms studied, LinRel and LinUCB, are similar to ours in that they use an upper confidence bound and both estimate the unknown parameters for the linear model using a least-square-error type method. In addition, both methods apply some sort of regularization. LinRel does it by truncating the eigenvalues of a certain matrix and LinUCB by using ridge regression. In the adversarial setting, and with no arm separation, the regret bounds obtained of the form $O(\sqrt{T} polylog(T))$. 

Dani \emph{et al.}  \cite{dani2008stochastic}, Rusmevichientong and Tsitsiklis \cite{rusmevichientong2008linearly}, and Abbasi-Yadkori \emph{et al.}~\cite{abbasi2011improved} study
 contextual linear bandits in the stochastic setting, in the case where $\myset{A}$ is a fixed but possibly uncountable bounded subset of $\reals^d$. 
 Dani \emph{et al.} \cite{dani2008stochastic}  obtain regret bounds of  $O(\sqrt{T})$ for an infinite number of arms; under arm separation, by introducing a gap constant $\Delta$, their bound is  $O(d^2 (\log T)^3)$.  Rusmevichientong and Tsitsiklis \cite{rusmevichientong2008linearly} also study the regret under arm separation and obtain a $O(\log(T))$ bound    that depends exponentially on $d$. Finally, Abbasi-Yadkori \emph{et al.}~\cite{abbasi2011improved} obtain a $O(poly(d)\log^2 (T))$ bound under arm separation.

Our problem can be expressed as a special case of the linear bandits setup by taking $\theta = [\theta_1;\ldots;\theta_K]\in\reals^{Kd}$, where $K=|\myset{A}|$, and, given context $x$, associating the $i$-th arm with an appropriate vector of the form $x_{a_i}=[0\ldots x \ldots 0]$. As such, all of the bounds described above \cite{auer2002tradeoffs,li2010contextual,chu2011contextual,dani2008stochastic,rusmevichientong2008linearly,abbasi2011improved}  can
be applied to our setup. However,  in our setting, arms are uncorrelated; the above algorithms do not exploit this fact. Our algorithm indeed exploits this to obtain a \emph{logarithmic} regret, while also scaling well in terms of the dimension $d$.


Several papers study contextual linear bandits under different notions of regret. For example, Dani \emph{et al.}~\cite{dani2007price}  define regret based on the worst sequence of loss vectors. In our setup, this corresponds to the rewards coming from an arbitrary temporal sequence and not from adding noise to $x^{\dagger} \theta_a$, resembling the `worst-case' regret definition of \cite{Auer95}. Abernethy \emph{et al.}~\cite{abernethy2008dark} assume a notion of regret with respect to a best choice fixed in time that the player can make from a fixed set of choices. However, in our case, the best choice changes with time $t$ via the current context. This different setup yields worse bounds  than the ones we seek: for both stochastic and adversarial setting the regret is $O(\sqrt{T}polylog(T))$.

Recent studies  on multi-class prediction using bandits~\cite{kakade2008,hazan2011,crammer2011} have some
connections to our work. In this setting, every context $x$ has an associated label $y$ that a learner tries to predict using a linear classifier of the type $\hat{y} = \arg \max_a \theta_a^{\dagger} x$.
Among algorithms proposed, the closest to ours is by Crammer and Gentile~\cite{crammer2011}, which uses an estimator for $\{\theta_a\}$ that is related to LinUCB, LinRel and our algorithm. However, the multi-class prediction problem differs in many ways from our setting. To learn the vectors $\theta_a$, the learner receives a one-bit feedback indicating whether the label predicted is correct (\emph{i.e.}, the arm was maximal) or not. In contrast, in our setting, the learner directly observes $\theta_a^{\dagger} x$, possibly perturbed by noise, without learning if it is maximal.

Finally, bandit algorithms relying on experts such as EXP4~\cite{auer2002exp4} and EXP4.P~\cite{beygelzimer2010contextual} can also be applied to our setting. These algorithms require a set of policies (experts) against which the regret is measured. Regret bounds grow as $\log^C N $, where $N$ is the number of experts and $C$ a constant. The trivial reduction of our problem to EXP4(.P) assigns an expert to each possible context-to-arm mapping. The $2^d$ contexts in our case lead to $K^{2^d}$ experts, an undesirable exponential growth of regret in $d$; a better choice of experts is a new problem in itself.

\section{Model}
\label{sec:setup}

In this section, we give a precise definition of our linear contextual bandit problem. 
 \paragraph{Contexts.}
At every time instant $t \in \{1,2, ... \}$, a context $x_t \in \X \subset \reals^d$, is observed by the learner. We assume that $\|x\|_2 \leq 1$; as the expected reward is linear in $x$, this assumption is without loss of generality (w.l.o.g.). We prove our main result (Theorem \ref{th:main_theorem}) in the stochastic setting where $x_t$ are drawn i.i.d.~from an unknown multivariate probability distribution $\mathcal{D}$. 
In addition, we require that the  set of contexts is finite \emph{i.e.}, $| \X | < \infty$.  We define $\Sigma_{\min} > 0$ to be the smallest non-zero eigenvalue of  the covariance matrix $\Sigma\equiv\E \{ x_1 x^{\dagger}_1\}$.  
\paragraph{Arms and Actions.}
At time $t$, after observing the context $x_t$, the learner decides to play an arm $a \in \A$, where $K \equiv |\A|$ is finite. We denote the arm played at this time by $a_t$. We study \emph{adaptive} arm selection policies, whereby the selection of $a_t$ depends only on the current context $x_t$, and on all past contexts, actions and rewards. In other words, $a_t = a_t\left(x_t,\{ x_{\tau}, a_{\tau}, r_{\tau} \}^{t-1}_{\tau = 1}\right)$.

\paragraph{Payoff.}
After observing a context $x_t$ and selecting an arm $a_t$, the learner receives a payoff $r_{a_t,x_t}$ which is drawn from a distribution $p_{a_t,x_t}$ independently of all past contexts, actions or payoffs. We assume that the expected payoff is a linear function of the context. In other words,
\begin{equation}
r_{a_t,x_t} = x^{\dagger}_t \theta_a + \epsilon_{a,t}
\end{equation} 
where $\{ \epsilon_{a,t} \}_{a \in \A, t \geq 1}$ are a set of independent random variables with zero mean and $\{ \theta_a\}_{a \in \A}$ are unknown parameters in $\reals^d$. 
Note that, w.l.o.g, we can assume that $Q=\max_{a\in\A}\| \theta_a \|_2\leq 1$. This is because if $Q>1$ , as payoffs are linear, we can divide all payoffs by $Q$; the resulting payoff is still a linear model, and our results stated below apply. Recall that $Z$ is a sub-gaussian random variable with constant $L$ if $\E \{ e^{\gamma Z }\} \leq e^{\gamma^2 L^2} $. In particular, sub-gaussianity implies $\E\{ Z\} = 0$.  We make the following technical assumption.
\begin{assump} \label{ass:first}
The random variables $\{ \epsilon_{a,t} \}_{a \in \A, t \geq 1}$ are sub-gaussian random variables with constant $L>0$. 
\end{assump}

\paragraph{Regret.}
Given a context $x$, the optimal arm is
$a^*_{x} = \argmax_{a \in \A} x^{\dagger} \theta_a.$
The expected cumulative regret the learner experiences over $T$ steps is defined by
\begin{equation}
R(T) = \E \Big \{ \sum^T_{t = 1} x_t^{\dagger} (\theta_{a^*_{x_t}} - \theta_{a_t}) \Big \}. 
\end{equation}
The expectation above is taken over the contexts $x_t$.
The objective of the learner is to design a policy $a_t = a_t\left(x_t,\{ x_{\tau}, a_{\tau}, r_{\tau} \}^{t-1}_{\tau = 1}\right)$ that achieves as low expected cumulative regret as possible. In this paper we are also interested in arm selection policies having a low computational complexity.
We define $\Delta_{\max} \equiv \max_{a,b\in \A} ||\theta_a-\theta_b||_2$, and  
$\Delta_{\min} \equiv \inf_{x \in \X, a: x^\dagger \theta_a<x^\dagger \theta_{a^*_x}} x^{\dagger} (\theta_{a^*_x} - \theta_a) > 0.$
Observe that, by the finiteness of $\X$ and $\A$, the defined infimum is attained (\emph{i.e.}, it is a minimum) and is indeed positive. 

\section{Main Results}
\label{sec:main_results}
We now present a simple and \emph{efficient} on-line algorithm that, under the above assumptions, has expected \emph{logarithmic} regret. Specifically, its computational complexity, at each time instant, is $O(K d^3 )$ and the expected memory requirement scales like $O(K d^2)$. As far as we know, our analysis is the first to show that a simple and \emph{efficient} algorithm for the problem of linearly parametrized bandits can, under reward separation and i.i.d. contexts, achieve logarithmic expected cumulative regret that
simultaneously can scale like $polylog(|\mathcal{X}|)$ for natural scenarios.

Before we present our algorithm in full detail, let us give some intuition about it. Part of the job of the learner is to estimate the unknown parameters $\theta_a$ based on past actions, contexts and rewards. We denote the estimate of $\theta_a$ at time $t$ by $\hat{\theta}_{a}$. If $\theta_a \approx \hat{\theta}_{a}$ then, given an observed context, the learner will more accurately know which arm to play  to incur in small regret. 
The estimates $\hat{\theta}_{a}$ can be constructed based on a history of past rewards, contexts and arms played. 
Since observing a reward $r$ for arm $a$ under context $x$ does not give information about the magnitude of $\theta_a$ along directions orthogonal to $x$, it is important that, for each arm, rewards are observed and recorded for a rich class of contexts. This gives rise to the following challenge:  If the learner tries to build this history while trying to minimize the regret, the distribution of contexts observed when playing a certain arm $a$ will be biased and potentially not rich enough. In particular, when trying to achieve a small regret, conditioned on $a_t = a$, it is more likely that $x_t$ is a context for which $a$ is optimal. 

We address this challenge using the following idea, also appearing in the epoch-Greedy algorithm of \cite{langford2007epoch}. We partition time slots into  \emph{exploration} and \emph{exploitation epochs}. In exploration epochs, the learner plays arms uniformly at random, independently of the context, and records the observed rewards. This guarantees that in the history of past events, each arm has been played along with a sufficiently rich set of contexts. In exploitation epochs, the learner makes use of the history of events stored during exploration to estimate the parameters $\theta_a$ and determine which arm to play given a current observed context. The rewards observed during exploitation are not recorded.

More specifically, when exploiting, the learner performs  two operations. In the first operation, for each arm $a \in \A$, an estimate $\hat{\theta}_{a}$ of $\theta_a$ is constructed from a simple $\ell_2$-regularized regression, as in in \cite{auer2002tradeoffs} and \cite{chu2011contextual}. In the second operation, the learner plays the arm $a$ that maximizes $x^{\dagger}_t \hat{\theta}_{a}$. Crucially, in the first operation, only information collected during exploration epochs is used.
  In particular, let $\T_{a,t-1}$ be the set of exploration epochs up to and including time $t-1$ (\emph{i.e.}, the times that the learner played an arm $a$ uniformly at random (u.a.r.)). Moreover, for any  $\T\subset \naturals$, denote by $r_{\T} \in \reals^n$ the vector of observed rewards for all time instances  $t\in\T$,   and $X_{\T} \in \reals^{n \times d}$ is a matrix of $\T$ rows, each  containing one of the observed contexts at time $t\in \T$.
 Then, at time  $t$
the estimator $\hat{\theta}_a$ is the solution of the following convex optimization problem. 
\begin{align}\label{minrloglik}
\min_{\theta \in \reals^d} \frac{1}{2 n} \| r_{\T} - X_{\T} \theta  \|^2_2 + \frac{\lambda_{n}}{2} \| \theta\|^2_2.
\end{align} 
where $\T = \T_{a,t-1}$, $n=|\T_{a,t-1}|$, $\lambda_{n} = 1/\sqrt{n}$. In other words, the estimator $\hat{\theta}_a$ is a (regularized) estimate of $\theta_a$, based only on observations made during exploration epochs.
Note that the solution to \eqref{minrloglik} is given by
$\hat{\theta}_a = \left(\lambda_nI+\frac{1}{n} X_{\T}^\dagger X_{\T}\right )^{-1} \frac{1}{n} X_{\T}^\dagger r_{\T}.$

\begin{algorithm}[!t]
\caption{Contextual $\epsilon$ -greedy}
\label{cegreedy}
\begin{algorithmic}
\STATE For all $a\in A$, set $A_a \leftarrow 0_{d\times d}$ ;$n_a\leftarrow 0$; $b_a\leftarrow 0_{d}$
\FOR{$t = 1$ to $p$}
\STATE $a \leftarrow 1 + (t \mod K)$; Play arm $a$
\STATE $n_a\leftarrow n_a+1$; $b_a \leftarrow b_a+r_tx_t$; $A_a\leftarrow A_a+ x_tx_t^\dagger$
\ENDFOR
\FOR{$t= p+1 $ to $T$}
\STATE $e \leftarrow \text{Bernoulli}(p/t)$
\IF{$e = 1$} 
\STATE $a \leftarrow \text{Uniform}(1/K)$ ; Play arm $a$
\STATE $n_a\leftarrow n_a+1$; $b_a \leftarrow b_a+r_tx_t$; $A_a\leftarrow A_a+ x_tx_t^\dagger$
\ELSE
\FOR{$a \in \A$}
\STATE Get $\hat{\theta}_a$ as the solution to the linear system: $\left( \lambda_{n_a} I + \frac{1}{n_a} A_a\right) \hat\theta_a =  \frac{1}{n_a} b_a$
\ENDFOR
\STATE Play arm $a_t = \arg \max_{a \in \myset{A}} x^{\dagger}_t \hat{\theta}_a$
\ENDIF
\ENDFOR 
\end{algorithmic}
\end{algorithm} 
An important design choice is the above process selection of the time slots at which the algorithm explores, rather than exploits.
Following the ideas of \cite{barto1998}, we select the exploration epochs so that they occur approximately $\Theta(\log t)$ times after $t$ slots. This guarantees that, at each time step, there is enough information in our history of past events to determine the parameters accurately while only incurring in a regret of $O(\log t)$. There are several ways of achieving this;  our algorithm explores at each time step with probability $\Theta(t^{-1})$. 

The above steps are summarized in pseudocode by Algorithm~\ref{cegreedy}. Note that the algorithm contains a scaling parameter $p$, which is specified below, in Theorem~\ref{th:main_theorem}. Because there are $K$ arms and for each arm  $(x_t, r_{a,t}) \in \reals^{d+1}$, the expected memory required by the algorithm scales like $O(K d^2)$. In addition, both the matrix $X^{\dagger}_{\T} X_{\T} $ and the vector $X^{\dagger}_{\T} r_{\T} $ can be computed in an online fashion in $O(d^2)$ time: $X^{\dagger}_{\T} X_{\T}  \leftarrow X^{\dagger}_{\T} X_{\T} + x_t x^{\dagger}_t$ and $X^{\dagger}_{\T} r_{\T} \leftarrow X^{\dagger}_{\T} r_{\T}  + r_{t} x_t $. Finally, the estimate of $\hat{\theta}_{a}$ requires solving a linear system (see Algorithm \ref{cegreedy}), which can be done in $O(d^3)$ time. The above is summarized in the following theorem.

\begin{theorem}
Algorithm \ref{cegreedy} has computational complexity of $O(K d^3)$ per iteration and its expected space complexity scales like $O(K d^2 )$.
\end{theorem}

We now state our main theorem that shows that Algorithm \ref{cegreedy} achieves $R(T)  = O(\log T)$.
\begin{theorem} \label{th:main_theorem}
Under Assumptions \ref{ass:first}, the expected cumulative regret of algorithm \ref{cegreedy} satisfies,
\begin{align*}
R(T) \leq  p \Delta_{\max} \sqrt{d} + 14 \Delta_{\max} \sqrt{d} K e^{Q/4} +  p \Delta_{\max} \sqrt{d} \log T.
\end{align*}
for any
\begin{equation}\label{eq:ineq_for_p}
p \geq  \frac{C K L'^2}{ (\Delta'_{\min})^2 (\Sigma'_{\min})^2}.
\end{equation}
Above, $C$ is a universal constant, $\Delta'_{\min} = \min\{1,\Delta_{\min} \}$, $\Sigma'_{\min} = \min \{1, \Sigma_{\min}\}$ and $L' = \max\{1, L\}$.
\end{theorem}

Algorithm \ref{cegreedy} requires the specification of the constant $p$.
In Section \ref{sec:computing_p}, we give two examples of how
to efficiently choose a $p$ that satisfies  \eqref{eq:ineq_for_p}.  
In Theorem \ref{th:main_theorem}, the bound on the regret
depends on $p$ - small $p$ is preferred - and
hence it is important to understand how the right hand side (r.h.s.) of \eqref{eq:ineq_for_p} might scale when $K$ and $d$ grow.
In Section \ref{sec:example_scale_p}, we show 
that, for a concrete distribution of contexts and choice of expected rewards $\theta_a$,
and assuming \eqref{eq:ineq_for_p} holds,
$p = O(K^3 d^5)$ \footnote{This bound holds with probability
converging to $1$ as $K$ and $d$ get large} .
There is nothing special
about the concrete details of how contexts and $\theta_a$'s are chosen
and, although not included in this paper,
for many other distributions, one also obtains $p = O(poly(d))$.
We can certainly construct pathological cases where,
for example, $p$ grows exponentially with $d$.
However, we do not find these intuitive. Specially when
interpreting these having in mind real applications as the ones
introduced in Examples 1- 3.

\vspace*{-0.3cm}
\subsection{Example of Scaling of $p$ with $d$ and $K$} \label{sec:example_scale_p}
\vspace*{-0.2cm}

Assume that contexts are obtained by normalizing
a $d$-dimensional vector with i.i.d. entries as Bernoulli random variables with parameter $w$.
Assume in addition that every $\theta_a$ is obtained i.i.d. from the
following prior distribution: every entry of $\theta_a$ is drawn i.i.d. from a uniform
distribution and then $\theta_a$ is normalized.
Finally, assume that the payoffs are given by $r_{a,t} = x^{
\dagger}_t \Theta_a$, where $\Theta_a \in \reals^d$ are random variables
that fluctuate around $\theta_a = \E \{ \Theta_a\}$ with
each entry fluctuating by at most $F$. 

Under these assumptions the following is true:
\begin{itemize}
\item  $\Sigma_{\min} = \Omega(d^{-1} )$. In fact, the same result
holds asymptotically independently of $w = w(d)$ if, for example, we assume that on
average groups are roughly of the same size, $M$, with $w = M/d $;
\item $L=O(\sqrt{d})$. This holds because $\epsilon_{a,t} = r_{a,t} - \E \{r_{a,t} \}= x^{\dagger}_t (\Theta_a - \theta_a)$ are bounded random variables with zero mean and $\| x^{\dagger}_t (\Theta_a - \theta_a) \}\|_{\infty} = O(\sqrt{d})$.
\item  $\Delta_{\min}  = \Omega(1/ (K d \sqrt{w})$ with high-probability (for large $K$ and $d$).
This can be seen as follows, if $\Delta_{\min} = x^{\dagger} (\theta_a - \theta_b)$
for some $x$, $a$ and $b$, then it must be true that $\theta_a$ and $\theta_b$ differ
in a component for which $x$ is non-zero. The minimum difference between
components among all pairs of $\theta_a$ and $\theta_b$ is lower bounded by
$\Omega(1/(K \sqrt{d}))$ with high probability (for large $K$ and $d$). Taking
into account that each entry of $x$ is   $O(1/\sqrt{d w})$ with high-probability,
the bound on $\Delta_{\min}$ follows.
\end{itemize}

If we want to apply Theorem \ref{th:main_theorem} then \eqref{eq:ineq_for_p} must
hold and hence putting all the above calculations together we conclude that $p = O(K^3 d^5)$ with high probability for large $K$ and $p$.
\vspace*{-0.3cm}
\subsection{Computing $p$ in Practice}\label{sec:computing_p}
\vspace*{-0.2cm}

If we have knowledge of an a priori distribution for the contexts,
for the expected payoffs and for the variance of the rewards then
we can quickly compute the value of $\Sigma_{\min}$, $L$
and a typical value for $\Delta_{\min}$. An example of this
was done above (Section \ref{sec:example_scale_p}). 
There, the values were presented only in order notation but
exact values are not hard to obtain for that and other distributions.
Since a suitable $p$ only needs to be larger then the r.h.s.
of \eqref{eq:ineq_for_p}, by introducing an appropriate multiplicative
constant, we can produce a $p$ that satisfied \eqref{eq:ineq_for_p} with
high probability.

If we have no knowledge of any model for the contexts or expected payoffs, it is still possible to
find $p$ by estimating $\Delta_{\min}$, $\Sigma_{\min}$ and $L$ from
data gathered while running Algorithm \ref{cegreedy}.
Notice again that, since all that is required for our theorem to hold is that $p$ is greater then a certain function of these quantities,
an exact estimation is not necessary. This is important because, for example, accurately estimating $\Sigma_{\min}$
is hard when matrix $\E \{x_1 x^{\dagger}_1 \}$ has a large condition number.

Not being too concerned about accuracy, $\Sigma_{\min}$ can be estimated from $\E \{x_1 x^{\dagger}_1 \}$, which can be estimated from the sequence of observed $x_t$.
$\Delta_{\min}$ can be estimated from Algorithm \ref{cegreedy} by keeping
track of the smallest difference observed until time $t$ between
$\max_b x^{\dagger} \hat{\theta}_b$ and the second largest value of the function being maximized.
Finally, the constant $L$ can be estimated from the variance of the observed rewards
for the same (or similar) contexts. Together, these estimations do not incur in any
significant loss in computational performance of our algorithm.

\section{Proof of Theorem \ref{th:main_theorem}}
\label{sec:proofs}
The general structure of the proof of our main result follows that of \cite{auer2002finite}. The main technical innovation is the realization that, in the setting when the contexts are drawn i.i.d.~from some distribution, a standard matrix concentration bound allows us to treat  $\lambda_{n} I +  n^{-1} (X^{\dagger}_{\T} X_{\T})$
 in Algorithm \ref{cegreedy}
as a deterministic positive-definite symmetric matrix, even as $\lambda_n \rightarrow 0$. 

Let $\mathcal{E}_T$ denote the time instances for $t > p$ and until time $T$ in which the algorithm took an exploitation decision.
Recall that, by Cauchy-Schwarz inequality, $x^{\dagger}_t (\theta_{a^*_{x_t}} - \theta_{a}) \leq \|x_t \|_1 \|(\theta_{a^*_{x_t}} - \theta_{a}) \|_{\infty} \leq \sqrt{d} \| x_t\|_2 \|(\theta_{a^*_{x_t}} - \theta_{a}) \|_{\infty}  \leq \sqrt{d} \Delta_{\max}$. In addition, recall that $\sum^T_{t=2} 1/t \leq \log T$.  For $R(T)$ the cumulative regret until time $T$, we can write
\begin{align*}
R(T) &= \E \big\{  \sum^T_{t = 1}  x^{\dagger}_t (\theta_{a^*_{x_t}} - \theta_{a}) \big\} \leq p \Delta_{\max} \sqrt{d}
  + \Delta_{\max} \sqrt{d} \E \big\{  \sum^T_{t = p+1}  \id{\{ x^{\dagger}_t {\theta}_{a} < x^{\dagger}_t {\theta}_{a^*_{x_t}}  \}} \big\}\\
 & \leq p \Delta_{\max} \sqrt{d} +  \Delta_{\max} \sqrt{d} \E \{ |\mathcal{E}_T| \}
+ \Delta_{\max} \sqrt{d} \E \big\{  \sum_{t \in \mathcal{E}_T}  \id{\{x^{\dagger}_t {\theta}_{a} <x^{\dagger}_t {\theta}_{a^*_{x_t}}  \}} \big\}\\
& \leq p \Delta_{\max} \sqrt{d} +  p \Delta_{\max} \sqrt{d} \log T
+ \Delta_{\max} \sqrt{d} \E \big\{  \sum_{t \in \mathcal{E}_T}  \id{\{ x^{\dagger}_t {\theta}_{a} < x^{\dagger}_t {\theta}_{a^*_{x_t}}  \}} \big\}\\
& \leq p \Delta_{\max} \sqrt{d} +  p \Delta_{\max} \sqrt{d} \log T
+ \Delta_{\max} \sqrt{d} \E \big\{  \sum_{t \in \mathcal{E}_T}  \sum_{a \in \A} \id{\{ x^{\dagger}_t \hat{\theta}_{a} >x^{\dagger}_t \hat{\theta}_{a^*_{x_t}}   \}} \big\}.
\end{align*}
In the last line we used the fact that when exploiting, if we do not exploit the optimal arm $a^*_{x_t}$, then it must be the case that the estimated reward for some arm $a$, $x^{\dagger}_t \hat{\theta}_{a}$, must exceed that of the optimal arm, $x^{\dagger}_t \hat{\theta}_{a^*_{x_t}}$, for the current context $x_t$. 

We can continue the chain of inequalities and write,
\begin{align*}
R(T) & \leq p \Delta_{\max} \sqrt{d} +  p \Delta_{\max} \sqrt{d} \log T
+ \Delta_{\max} \sqrt{d} K \sum^T_{t = 1} \prob \{ x^{\dagger}_t \hat{\theta}_{a} >x^{\dagger}_t \hat{\theta}_{a^*_{x_t}} \}.
\end{align*}
The above expression depends on the value of the estimators for time instances that might or might not be exploitation times. For each arm, these are computed just like in Algorithm \ref{cegreedy}, using the most recent history available. The above probability depends on the randomness of $x_t$ and on the randomness of recorded history for each arm.

Since $x^{\dagger}_t (\theta_{a^*_{x_t}} - \theta_{a}) \geq \Delta_{\min}$ we can write
\begin{align*}
&\prob \{ x^{\dagger}_t \hat{\theta}_{a} > x^{\dagger}_t \hat{\theta}_{a^*_{x_t}} \} 
\leq \prob \Big \{ x^{\dagger}_t \hat{\theta}_{a} \geq x^{\dagger}_t {\theta}_{a} + \frac{\Delta_{\min}}{2} \Big \} 
+ \prob \Big \{ x^{\dagger}_t \hat{\theta}_{a^*_{x_t}} \leq x^{\dagger}_t {\theta}_{a^*_{x_t}} - \frac{\Delta_{\min}}{2} \Big\}.
\end{align*}

We now bound each of these probabilities separately. Since their bound is the same, we focus only on the first probability.

Substituting the definition of $r_{a}(t) = x^{\dagger}_t \theta_a + \epsilon_{a,t}$ into the expression for $\hat{\theta}_{a}$ one readily obtains,
\begin{align*}
&( \hat{\theta}_{a} - \theta_{a} ) = \left( \lambda_n I + \frac{1}{n} X^\dagger_{\T} X_{\T} \right)^{-1} 
  \left( \frac{1}{n} \sum_{\tau \in \T} x_\tau  \epsilon_{a,\tau } - \lambda_{n} \theta_{a} \right).
\end{align*}
We are using again the notation $\T = \Tatm$ and $n = |\T|$. From this expression, an application of Cauchy-Schwarz's inequality and the triangular inequality leads to,
\begin{align*}
& |x_{t}^\dagger( \hat{\theta}_{a} - \theta_{a} ) |= \Big\lvert x_{t}^\dagger \left( \lambda_n I + \frac{1}{n} X^\dagger_{\T} X_{\T} \right)^{-1} 
  \left( \frac{1}{n} \sum_{\tau \in \T} x_\tau \epsilon_{a,\tau} - \lambda_n \theta_{a} \right) \Big\rvert\\
&\leq \sqrt{x_{t}^\dagger  \left( \lambda_n I + \frac{1}{n} X^\dagger_{\T} X_{\T} \right)^{-2} x_{t}} 
\left( \Big | \frac{1}{n} \sum_{\tau \in \T} x_t^\dagger x_\tau \epsilon_{a,\tau}  \Big | + \lambda_{n} |x_t^\dagger \theta_{a}| \right).
\end{align*}
We introduce the following notation
\begin{equation}
c_{a,t} \equiv \sqrt{x_{t}^\dagger  \left( \lambda_n I + \frac{1}{n} X^\dagger_{\T} X_{\T} \right)^{-2} x_{t}}.
\end{equation}
Note that, given $a$ and $t$ both $n$ and $\T$ are well specified.

We can now write,
\begin{align*}
&\prob \Big \{ x^{\dagger}_t \hat{\theta}_{a} \geq x^{\dagger}_t {\theta}_{a} + \frac{\Delta_{\min}}{2} \Big\}
\leq \prob \Big\{ \Big | \frac{1}{n} \sum_{\tau \in \T} x_t^\dagger x_\tau \epsilon_{a,\tau}  \Big | \geq  \frac{\Delta_{\min}}{2 c_{a,t}} -  \lambda_{n} |x_t^\dagger \theta_{a}| \Big\}\\
&\leq \prob \Big\{ \Big | \frac{1}{n} \sum_{\tau \in \T} x_t^\dagger x_\tau \epsilon_{a,\tau}  \Big | \geq  \frac{\Delta_{\min}}{2 c_{a,t}} -  \lambda_{n} Q \Big\}.
\end{align*}
Since $\epsilon_{a,\tau}$ are sub-gaussian random variables with sub-gaussian constant upper bounded by $L$ and since $|x^{\dagger}_t x_{\tau}| \leq 1$, conditioned on $x_t$, $\T$ and $\{x_{\tau}\}_{\tau \in \T}$, each $x^{\dagger}_t x_{\tau} \epsilon_{a,\tau}$ is a sub-gaussian random variable and together they form a set of i.i.d. sub-gaussian random variables. One can thus apply standard concentration inequality and obtain,
\begin{align}
\begin{split}
&\prob \Big\{ \Big | \frac{1}{n} \sum_{\tau \in \T} x_t^\dagger x_\tau \epsilon_{a,\tau}  \Big | \geq  \frac{\Delta_{\min}}{2 c_{a,t}} -  \lambda_{n} Q \Big\} 
 \leq \E \Big \{ 2 e^{- \frac{n}{2L^2} {\left( \frac{\Delta_{\min}}{2 c_{a,t}} -  \lambda_{n} Q \right)^+}^2 } \Big \}.
\end{split}
\label{eq:prob_err_bound_1}
\end{align}
where both $n$ and $c_{a,t}$ are random quantities and  $z^+ = z$ if $z \geq 0$ and zero otherwise.

We now upper bound $c_{a,t}$ using the following fact about the eigenvalues of any two real-symmetric matrices $M_1$ and $M_2$: $\lambda_{\max}(M_1^{-1}) = 1/\lambda_{\min}(M_1)$ and $\lambda_{\min}(M_1 + M_2) \geq \lambda_{\min}(M_1) - \lambda_{\max}(M_2) = \lambda_{\min}(M_1) - \|M_2\|$.
\begin{align*}
&c_{a,t} 
 \leq \left( \lambda_n + \lambda^+_{\min}(\E \{x^{\dagger}_1 x_1 \}) - \Big \|\frac{1}{n} X^\dagger_{\T} X_{\T} - \E \{x^{\dagger}_1 x_1 \} \Big \|^+ \right)^{-1}.
\end{align*}
Both the eigenvalue and the norm above only need to be computed over the subspace spanned by the vectors $x_t$ that occur with non-zero probability. We use the symbol ${}^{+}$ to denote the restriction to this subspace. Now notice that $\|.\|^{+} \leq \|.\|$ and, since we defined $\Sigma_{\min} \equiv \min_{i:\lambda_i > 0} \lambda_i (\E \{X_1 X^\dagger_1 \})$, we have that $\lambda^{+}_{\min}(\E \{X_1 X^{\dagger}_1 \}) \geq \Sigma_{\min}$. Using the following definition,
$\Delta \Sigma_n \equiv  n^{-1}   X^{\dagger}_{\T} X^{\dagger}_{\T} -  \E \{X_1 X^\dagger_1 \}$,
this leads to,
$c_{a,t} \leq ( \lambda_n + \Sigma_{\min} - \| \Delta \Sigma_n \| )^{-1} \leq \left( \Sigma_{\min} - \| \Delta \Sigma_n \| \right)^{-1}$.

We now need the following Lemma.
\begin{lemma}\label{th:matrix_conc_bound}
Let $\{X_i\}^n_{i=1} $ be a sequence of i.i.d. random vectors of 2-norm bounded by 1. Define $\hat{\Sigma} = \frac{1}{n} \sum^n_{i = 1} X_i X^{\dagger}_i$ and $\Sigma = \E \{ X_1 X^{\dagger}_1\} $. If $\epsilon \in (0,1)$ then,
\begin{equation*}
\prob(|\ \hat{\Sigma} - \Sigma \| > \epsilon \| \Sigma \| ) \leq 2 e^{-C \epsilon^2 n},
\end{equation*}
where $C < 1$ is an absolute constant.
\end{lemma}
For a proof see \cite{Vershynin2012} (Corollary 50).

We want to apply this lemma to produce a useful bound on the r.h.s.~of \eqref{eq:prob_err_bound_1}. First notice that, conditioning on $n$, the expression inside the expectation in \eqref{eq:prob_err_bound_1} depends through $c_{a,t}$ on $n$ i.i.d. contexts that are distributed according to the original distribution. Because of this, we can write,
\begin{align*}
&\prob \Big\{ \Big | \frac{1}{n} \sum_{\tau \in \T} x_t^\dagger x_\tau \epsilon_{a,\tau}  \Big | \geq  \frac{\Delta_{\min}}{2 c_{a,t}} -  \lambda_{n} Q \Big\} 
 \leq \E \Big \{ 2 e^{- \frac{n}{2L^2} {\left( \frac{\Delta_{\min}}{2 c_{a,t}} -  \lambda_{n} Q \right)^+}^2 } \Big \}\\
&\leq  \sum^t_{n = 1} \Big( \prob\{ | \Tatm| = n\}\nonumber 
\times \E \Big \{ 2 e^{- \frac{n}{2L^2} {\left( \frac{\Delta_{\min}}{2 c_{a,t}} -  \lambda_{n} Q \right)^+}^2 }  \Big |  | \Tatm| = n  \Big \} \Big ).\nonumber
\end{align*}

Using the following algebraic relation: if $z,w > 0$ then ${(z-w)^+}^2 \geq z^2 - 2 zw$, we can now write,
\begin{align*}
&\E \Big \{ e^{- \frac{n}{2 L^2} {\left( \frac{\Delta_{\min}}{2 c_{a,t}} -  \lambda_{n} Q \right)^+}^2 }  \Big |  | \Tatm| = n  \Big \}\\
&\leq \prob\{|\Delta \Sigma_n| > \Sigma_{\min}/2 | \; | \Tatm| = n \}
+ e^{- \frac{n}{2L^2} {\left( \frac{\Sigma_{\min} \Delta_{\min}}{4} -  \lambda_{n} Q \right)^+}^2 } \\
&\leq \prob\{|\Delta \Sigma_n| > \Sigma_{\min}/2 | \; | \Tatm| = n \}
 +e^{\frac{Q \Delta_{\min} \Sigma_{\min}}{4L^2}}  e^{- \frac{n (\Delta_{\min})^2 (\Sigma_{\min})^2 }{32L^2} } 
\end{align*}

Using Lemma \ref{th:matrix_conc_bound} we can continue the chain of inequalities,
\begin{align*}
&\E \Big \{ e^{- \frac{n}{2L^2} {\left( \frac{\Delta_{\min}}{2 c_{a,t}} -  \lambda_{n} Q \right)^+}^2 }  \Big |  | \Tatm| = n  \Big \}
  \leq  2e^{-C (\Sigma_{\min})^2 n/4}
 +e^{\frac{Q \Delta_{\min} \Sigma_{\min}}{4L^2}}  e^{- \frac{n (\Delta_{\min})^2 (\Sigma_{\min})^2 }{32 L^2} }.
\end{align*}
Note that $||\Sigma|| \leq 1$ follows from our non-restrictive assumption that $\|x_t\|_2 \leq1$ for all $x_t$.   
Before we proceed we need the following lemma:
\begin{lemma} \label{th:bound_on_n}
If $n_c = \frac{p}{2k} \log t$ , then
$\prob \{ | \Tatm | < n_c \} \leq t^{-\frac{p}{16K}}.$
\end{lemma}
\begin{proof}First notice that $|\T_{a,t-1} | = \sum^{t-1}_{i = 1} z_i$ where $\{z_i\}^{t-1}_{i =1}$ are independent Bernoulli random variables with parameter $p/(Ki)$. Remember that we can assume that $i > p$ since in the beginning of Algorithm \ref{cegreedy} we play each arm $p/K$ times.

Note that $\prob(X > c) \leq \prob(X + q > c)$ is always true for any r.v. $X,c$ and $q>0$. Now write,
\begin{align}
 \prob(|\T_{a,t-1}| < n_c )  &= \prob \left( \sum^{t-1}_{i = 1} z_i < n_c \right)\nonumber
= \prob \left( \sum^{t-1}_{i = 1} (z_i - p/(Ki))  < n_c - (p/K) \sum^{t-1}_{i = 1} 1/i \right)\nonumber\\
& \leq \prob \left( \sum^{t-1}_{i = 1} (-z_i + p/i)  > - n_c + (p/K) \sum^{t-1}_{i = 1} 1/i  \right)\nonumber\\
&\leq \prob \left( \sum^{t-1}_{i = 1} (-z_i + p/i)  > (p/K) \log t- n_c \right). 
\label{eq:interm_bernstein}
\end{align}
Since
$\sum^{t-1}_{i=1} \E \{(z_i - p/(Ki))^2\} = \sum^{t-1}_{i=p+1} (1 - p/(Ki))(p/(Ki)) \leq \frac{p}{K} \log t$,
we have that $\{ -z_i + p/i \}^{t-1}_{i =1} $ are i.i.d.~random variables with zero mean and sum of variances upper bounded by $(p/K) \log t$. Replacing $n_c  = (p/{2 K}) \log t $ in \eqref{eq:interm_bernstein} and applying Bernstein inequality we get,
%
$\prob(|\T_{a,t-1}| < n_c ) \leq e^{-\frac{\frac{1}{2} (p/(2K))^2 \log^2 t}{\frac{p}{K} \log t + \frac{1}{3} (p/(2K)) \log t}} \leq  t^{-\frac{p}{16 K}}.$\qed
\end{proof}
We can now write, by splitting the sum in $n< n_c$ and $n \geq n_c$
\begin{align*}
&\prob \Big\{ \Big | \frac{1}{n} \sum_{\tau \in \T} x_t^\dagger x_\tau \epsilon_{a,\tau}  \Big | \geq  \frac{\Delta_{\min}}{2 c_{a,t}} -  \lambda_{n} Q \Big\} \\
& \leq \sum^t_{n = 1} \prob\{ | \Tatm| = n\} 
\E \Big \{ 2 e^{- \frac{n}{2 L^2} {\left( \frac{\Delta_{\min}}{2 c_{a,t}} -  \lambda_{n} Q \right)^+}^2 }  \Big |  | \Tatm| = n  \Big \}\\
& \leq  \prob\{ | \Tatm| <  n_c\} + 4e^{-C (\Sigma_{\min})^2 n_c /4}
 +2e^{\frac{Q \Delta_{\min} \Sigma_{\min}}{4 L^2}}  e^{- \frac{n_c (\Delta_{\min})^2 (\Sigma_{\min})^2 }{32L^2} }\\
& \leq  t^{-\frac{p}{16 K}} + 4 t^{- \frac{C p (\Sigma_{\min})^2}{8K} } +2e^{\frac{Q \Delta_{\min} \Sigma_{\min}}{4L^2}}  t^{- \frac{p (\Delta_{\min})^2 (\Sigma_{\min})^2 }{64 K L^2} }.
\end{align*}
We want this quantity to be summable over $t$. Hence we require that,
\begin{align}\begin{split} 
p \geq \frac{128 K L^2}{(\Delta_{\min})^2 (\Sigma_{\min})^2},
p \geq \frac{16 K}{C (\Sigma_{\min})^2}, p \geq & 32 K.
\end{split}
\label{eq:cond_p}
\end{align}
It is immediate to see that our proof also follows if $\Delta_{\min}$, $\Sigma_{\min}$ and $L$ are replaced by $\Delta'_{\min} = \min\{1,\Delta_{\min} \}$, $\Sigma'_{\min} = \min \{1, \Sigma_{\min}\}$ and $L' = \max\{1, L\}$ respectively. If this is done, it is easy to see that conditions \eqref{eq:cond_p} are all satisfied by the $p$ stated in Theorem \ref{th:main_theorem}. 
Since $\sum^{\infty}_{t = 1} 1/t^2 \leq 2$, gathering all terms together we  have,
\begin{align*}
&R(T) \leq p \Delta_{\max} \sqrt{d} + p \Delta_{\max} \sqrt{d} \log T
+ \Delta_{\max} \sqrt{d} K \left( 4 e^{\frac{Q \Delta'_{\min} \Sigma'_{\min}}{4 L'^2}} + 10 \right)\\
& \leq p \Delta_{\max} \sqrt{d} + 14 \Delta_{\max} \sqrt{d} K e^{Q/4}+  p \Delta_{\max} \sqrt{d} \log T. &&~\qed
\end{align*}

\section{Numerical Results}\label{sec:numerical}

In Theorem~\ref{th:main_theorem}, we showed that, in the stochastic setting, Algorithm \ref{cegreedy} has an expected regret of  $O(\log T)$. We now illustrate this point by numerical simulations and, most importantly, exemplify how violating the stochastic assumption might degrade its performance.
Figure~\ref{figs} (a) shows the average cumulative regret (in semi-log scale) over 10 independent runs of Algorithm \ref{cegreedy} for $T = 10^5$ and for the following setup. The context variables $x \in \reals^3$ and at each time step $\{x_t\}_{t \geq 1}$ are drawn i.i.d. in the following way: (a) set each entry of $x$ to 1 or 0 independently with probability 1/2; (b) normalize $x$. We consider $K = 6$ arms with corresponding parameters $\theta_a$ generated independently from a standard multivariate gaussian distribution. Given a context $x$ and an arm $a$, rewards were random and independently generated from a uniform distribution $U([0, 2x^{\dagger} \theta_a])$. As expected, the regret is logarithmic.  Figure~\ref{figs} (a) shows a straight line at the end. 

To understand the effect of the stochasticity of $x$ on the regret, we consider the following scenario: with every other parameter unchanged, let $\X = \{x,x'\}$. At every time step $x = [1,1,1]$ appears with probability $1/I$, and $x' = [1,0,1]$ appears with probability $1-(1/I)$. Figure~\ref{figs} (b) shows the dependency of the expected regret on the context distribution for $I =$ 5, 10 and 100. One can see that an increase of $I$ causes a proportional increase in the regret. 
\begin{figure}[t]
\centering
\begin{tabular}{cc}
\small{(a)}  & \small{(b)}  \\
\!\!\!\!\!\!\includegraphics[width=2.4in,height=2.1in]{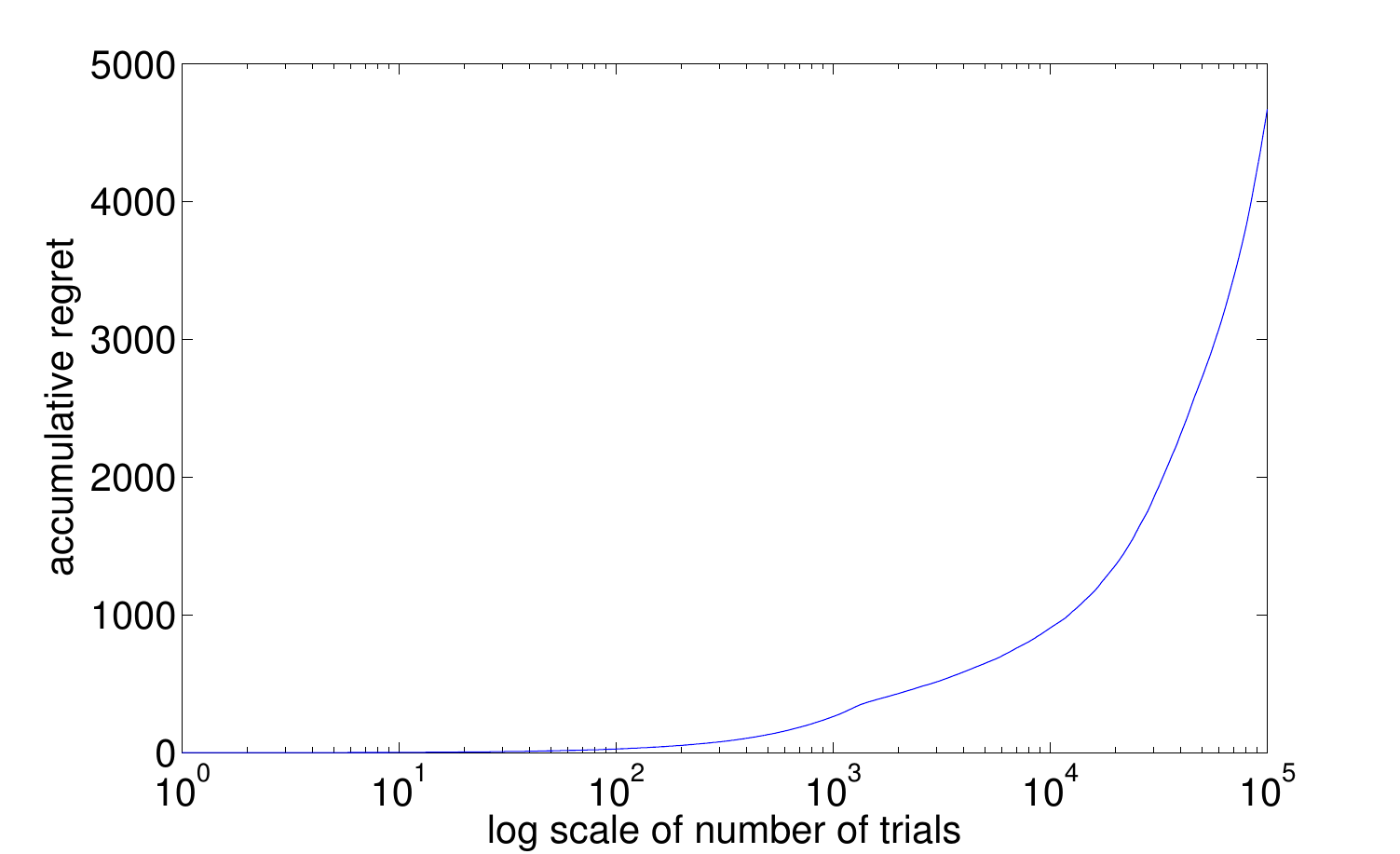}&
\!\!\!\!\!\!\includegraphics[height=2.1in]{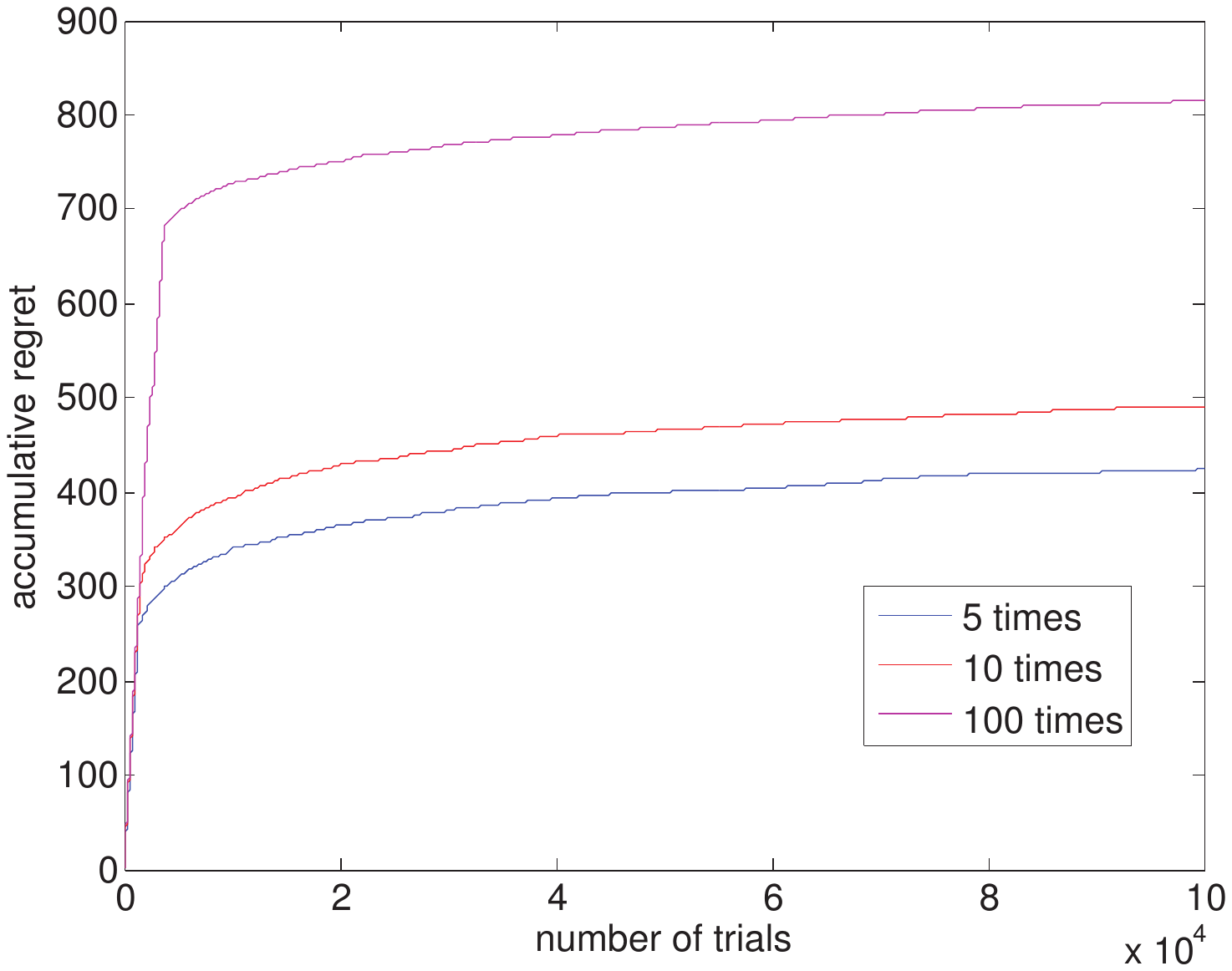}
\end{tabular}
\vspace*{-0.5cm}
\caption{\small{(a) Regret over T when $x_t$ is from i.i.d.  (b) Regret over T when $x_t$ is not from i.i.d.}\vspace*{-.5cm}}
\vspace*{-0.5cm}
\label{figs}
\end{figure}
 
 
 \vspace*{-0.3cm}
\section{Adversarial Setting}\label{sec:discussion}
\vspace*{-0.3cm}
 
In the stochastic setting, the richness of the subset of $\reals^d$ spanned by the observed contexts is related to the skewness of the distribution $\mathcal{D}$. The fact that the bound in Theorem~\ref{th:main_theorem} depends on $\Sigma_{\min}$ and that the regret increases as this value becomes smaller indicates that our approach does not yield a $O(\log T)$ regret for the adversarial setting, where an adversary choses the contexts and can, for example, generate $\{x_t\}$ from a
sequence of stochastic processes with decreasing $\Sigma_{\min}(t)$.

In particular, the main difficulty in using a linear regression, and the reason why our result depends on $\Sigma_{\min}$, is related to the dependency of our estimation of $x^{\dagger}_t \theta_a$ on $\frac{1}{|\Tatm| } X^{\dagger}_{\Tatm} X_{\Tatm}$. It is not hard to show that the error in approximating  $x^{\dagger}_t \theta_a$ with  $x^{\dagger}_t \hat{\theta}_a$ is proportional to
\begin{equation}\label{eq:path_mean}
\sqrt { x^{\dagger}_t \left(  \lambda_n I  + \frac{1}{n}  X^{\dagger}_{\T} X_{\T} \right)^{-2} x_t}.
\end{equation}
This implies that, even if a given context has been observed relatively often in the past, the algorithm can ``forget'' it because of the \emph{mean} over contexts that is being used to produce estimates of $x^{\dagger}_t \theta_a$ (the mean shows up
in \eqref{eq:path_mean} as $\frac{1}{n}  X^{\dagger}_{\T} X_{\T}$). 

The effect of this phenomenon on the performance of Algorithm~\ref{cegreedy} can be readily seen in the following pathological example. Assume that $ \X  = \{  (1, 1), (1, 0) \} \subset \reals^2$. Assume that the contexts arrive in the following way: $(1, 1)$ appears with probability $1/I$ and $(1, 0)$ appears with probability $1 - 1/I$. The correlation matrix for this stochastic process is $\{  (1, 1/I), (1/I, 1/I ) \}$ and its minimum eigenvalue scales like $O(1/I)$. Hence, the regret scales as $O (I^2 \log T)$. If $I$ is allowed to slowly grow with $t$, we expect that our algorithm will not be able to guarantee a logarithmic regret (assuming that our upper bound is tight). In other words, although $(1, 1)$ might have appeared a sufficient number of times for us to be able to predict the expected reward for this context, Algorithm~\ref{cegreedy} performs poorly since the mean \eqref{eq:path_mean} will the `saturated' with the context $(1, 0)$ and forget about $(1, 1)$.

One solution for this problem is to ignore some past contexts when building an estimate for $x^{\dagger}_t \theta_a$, by including in the mean \eqref{eq:path_mean} past contexts that are closer in direction to the current context $x_t$. Having this in mind, and building on the ideas of~\cite{auer2002tradeoffs}, we propose the UCB-type Algorithm \ref{cucb}.
\begin{algorithm}[!t]
\caption{Contextual UCB}
\label{cucb}
\begin{algorithmic}
\FOR{$t = 1$ to $p$}
\STATE $a \leftarrow 1 + (t \mod K)$; Play arm $a$; $\Tat \leftarrow \Tatm \cup \{t\}$
\ENDFOR
\FOR{$t= p+1 $ to $T$}
\FOR{$a \in \A$}
\STATE
$c_{a,t} \leftarrow \displaystyle{\min_{\T \subset \Tatm}} \frac{\log t}{| \T |} x^{\dagger}_t \left(  \lambda_n I  + \frac{1}{n}  X^{\dagger}_{\T} X_{\T} \right)^{-2} x_t$
\STATE $\T^* \leftarrow$ subset of $\Tatm$ that achieves the minimum; $n \leftarrow |\T^*|$
\STATE
Get $\hat{\theta}_a$ as the solution to the linear system: $\left( \lambda_n I + \frac{1}{n} X^{\dagger}_{\T} X_{\T}  \right) \hat{\theta}_a  =  \left( \frac{1}{n} X^{\dagger}_{\T} r_{\T} \right)$
\ENDFOR
\STATE Play arm $a_t = \arg \max_a x^{\dagger}_t \hat{\theta}_a +  \sqrt{c_{a,t}}$; Set $\Tat \leftarrow \Tatm \cup \{t\}$
\ENDFOR 
\end{algorithmic}
\end{algorithm}
\vspace*{+0.3cm}

It is straightforward to notice that this algorithm cannot be implemented in an efficient way. In particular, the search for $\T^* \subset \Tatm$ has a computational complexity exponential  in $t$. The challenge is to find an efficient way of approximating $\T^*$ efficiently. This can be done by either reducing the size of $\Tatm$ -- the history from which one wants to extract $\Tatm$ -- by not storing all events in memory (for example, if we can guarantee that $|\Tat| = O(\log t)$ then the complexity of the above algorithm at time step $t$ is $O(t)$), or by finding an efficient algorithm of approximating the minimization over the $\Tatm$ (or both).  It remains an open problem to find such an approximation scheme and to prove that it achieves $O(\log T)$ regret for a setting more general than the i.i.d. contexts considered in this paper.

\vspace*{-0.2cm}
\section{Conclusions}
\label{sec:conclusion}
\vspace*{-0.3cm}

We introduced an $\epsilon$-greedy type of algorithm that provably achieves logarithmic regret 
for the contextual multi-armed bandits problem with linear payoffs in the stochastic setting.
Our online algorithm is both fast and uses small space. 
In addition, our bound on the regret scales nicely with dimension of the contextual variables, $O(poly(d) \log T)$.
By means of numerical simulations we illustrate how the stochasticity of the contexts is important for our bound to hold.
In particular, we show how to construct a scenario for which our algorithm does not give logarithmic regret.
The reason for this amounts to the fact that the mean $n^{-1} X^{\dagger}_\T X_\T$ that is used in estimating the
parameters $\theta_a$ can ``forget'' previously observed contexts. Because of this, it remains an open problem
to show that there are efficient algorithms that achieve $O(poly(d) \log T)$ under reward separation ($\Delta_{\min} > 0$) in the non-stochastic setting. We believe that a possible solution might be constructing a variant of our algorithm where in $n^{-1} X^{\dagger}_\T X_\T$ we use a more careful average of past observed contexts give the current observed context. In addition, we leave it open to produce simple and efficient online algorithms for multi-armed bandit problems under rich context models, like the one we have done here for linear payoff. 
 
\vspace*{-0.3cm}
\section{Acknowledgement}
\vspace*{-0.3cm}
This work was sponsored by the NSF Grant 1161151: AF: Sparse Approximation: Theory and Extensions.

\eat{

}
\bibliographystyle{splncs03}
\bibliography{bandit}

 \end{document}